\documentclass[11pt]{article}
\usepackage[T1]{fontenc}
\usepackage{lmodern}
\usepackage{fullpage}
\usepackage{authblk}
\usepackage{amssymb}
\usepackage[english]{babel}
\usepackage[utf8]{inputenc}
\usepackage{cancel}
\usepackage{xspace}
\usepackage{bm}
\usepackage{float}
\usepackage{microtype}
\usepackage{amsthm}
\usepackage{amsmath}
\usepackage{amsfonts}
\usepackage{graphicx}
\usepackage[colorinlistoftodos]{todonotes}
\usepackage{mathtools}
\usepackage[ruled,vlined,linesnumbered]{algorithm2e}

\usepackage{tikz}
\usepackage{pgfplots}
\pgfplotsset{compat = newest}

\usepackage{color}
\definecolor{darkgreen}{rgb}{0,0.5,0}
\definecolor{darkblue}{rgb}{0,0,0.8}
\definecolor{darkred}{rgb}{0.8,0,0}

 \usepackage{hyperref}
 \hypersetup{
   unicode=false,          % non-Latin characters in Acrobat’s bookmarks
   colorlinks=true,        % false: boxed links; true: colored links
   linkcolor=darkred,          % color of internal links (change box color with linkbordercolor)
   citecolor=darkblue,        % color of links to bibliography
   filecolor=magenta,      % color of file links
   urlcolor=blue           % color of external links
 }
\newtheorem{definition}{Definition}[section]
\newtheorem{lemma}[definition]{Lemma}
\newtheorem{theorem}[definition]{Theorem}

\newtheorem{property}[definition]{Property}

\graphicspath{{./figures/}}

\newcommand{\bigo}{\mathcal{O}}

\newcommand\defeq{\stackrel{\mathclap{\normalfont{\mbox{\text{\tiny{def}}}}}}{=}}

\DeclareMathOperator*{\E}{\mathbb{E}}
\DeclareMathOperator*{\Var}{\mathrm{Var}}

\newcommand{\LL}{\textsf{LogLog}\xspace}
\newcommand{\HLL}{\textsf{HyperLogLog}\xspace}

 \title{\bf\LARGE Cardinality estimation using Gumbel distribution.}

 \author{\LARGE Aleksander~Łukasiewicz}
 \author{Przemys\l{}aw~Uzna\'nski}
\affil{\large Institute of Computer Science, University of Wrocław, Poland.}
 \date{}
\begin{document}
\maketitle

% % \setcounter{page}{0}
 \thispagestyle{empty}

% % \setcounter{page}{0}
% \thispagestyle{empty}

%\linenumbers
\begin{abstract}
\emph{Cardinality estimation} is the task of approximating the number of distinct elements in a large dataset with possibly repeating elements.  \LL and \HLL (c.f. Durand and Flajolet [ESA 2003], Flajolet et al. [Discrete Math Theor. 2007]) are small space sketching schemes for cardinality estimation, which have both strong theoretical guarantees of performance and are highly effective in practice. 
This makes them a highly popular solution with many implementations in big-data systems (e.g. Algebird, Apache DataSketches, BigQuery, Presto and Redis).
However, despite having simple and elegant formulation, both the analysis of \LL and \HLL are extremely involved -- spanning over tens of pages of analytic combinatorics and complex function analysis.

 We propose a modification to both \LL and \HLL that replaces discrete geometric distribution with a continuous Gumbel distribution. This leads to a very short, simple and elementary analysis of estimation guarantees, and smoother behavior of the estimator. 
\end{abstract}

\section{Introduction.}
In cardinality estimation problem we are presented with a dataset consisting of many items, that might be repeating. 
Our goal is to process this dataset efficiently, to estimate the number $n$ of \emph{distinct} elements it contains. Here, efficiently means in small auxiliary space, and fast processing per each item. 
A natural scenario to consider is a \emph{stream} processing of a dataset, with stream of events being either element \emph{insertions} to the multiset and \emph{queries} of multiset cardinality.

A folklore information theoretic analysis reveals that this problem over universe of $u$ elements requires at least $u$ bits of memory to answer queries exactly. However, in many practical settings it suffices to provide an approximate of the cardinality. An example scenario is estimating number of unique addresses in packets that a router observes, in order to detect malicious behaviors and attacks. Here limited computational capabilities of the router and sheer volume of data observed over e.g. day ask for specialized solutions.

The theoretical study of this problem was initiated by seminal work of Flajolet and Martin \cite{DBLP:journals/jcss/FlajoletM85}. Two follow-up lines of research follow. First, we mention
\cite{DBLP:conf/stoc/AlonMS96,
DBLP:conf/random/Bar-YossefJKST02,
DBLP:conf/soda/Bar-YossefKS02,
DBLP:conf/soda/Blasiok18,
DBLP:conf/vldb/Gibbons01,
DBLP:conf/spaa/GibbonsT01,
DBLP:conf/pods/KaneNW10} on the upper-bound side and  
\cite{DBLP:conf/stoc/AlonMS96,
DBLP:conf/coco/BrodyC09,
DBLP:conf/focs/IndykW03,
DBLP:conf/soda/JayramW11,
DBLP:conf/soda/Woodruff04} on lower-bound side. 
Those works focus on $(\varepsilon,\delta)$-guarantees, meaning that they guarantee outputting $(1+\varepsilon)$-multiplicative approximation of the number of distinct elements, with probability at least $1-\delta$. 
The high-level takeaway message is that one can construct approximate schemes that provide $(1+\varepsilon)$-multiplicative approximation to the number of distinct elements, using an order of $\varepsilon^{-2}$ space, and that this dependency on $\varepsilon$ is tight. 
More specifically, the work of Błasiok~\cite{DBLP:conf/soda/Blasiok18} settles the bit-complexity of the problem, by providing $\bigo(\frac{\log \delta^{-1}}{\varepsilon^2} + \log n)$ bits of space upper-bound, and this complexity is optimal by a matching lowerbound \cite{DBLP:conf/soda/JayramW11}.
To achieve such small space usage, a number of issues have to be resolved, and a very sophisticated machinery of expanders and pseudo-randomness is deployed.

The other line of work is more practical in nature, and focuses on providing variance bounds for efficient algorithm. The bounds are usually of the form $\sim 1/\sqrt{k}$ where $k$ is some measure of space-complexity of algorithms (usually, corresponds to the number of parallel estimation processes). This includes work of 
\cite{beyer2009distinct,
chen2011distinct,
cohen2015all,
DBLP:conf/esa/DurandF03,
DBLP:journals/ton/EstanVF06,
flajolet2007hyperloglog,
gerin2006efficient,
DBLP:journals/dam/Giroire09,
lumbroso2010optimal,
DBLP:journals/corr/abs-2007-08051,
DBLP:conf/kdd/Ting14,
viola2012data}. We now focus on two specific algorithms,  namely \LL \cite{DBLP:conf/esa/DurandF03} and later refined to \HLL \cite{flajolet2007hyperloglog}. The guarantees provided for variance are approximately $1.3/\sqrt{k}$ and $1.04/\sqrt{k}$ respectively, when using $k$ integer registers.
 Both are based on simple principle of observing the maximal number of trailing zeroes in binary representation of hashes of elements in the stream, although they vary in the way they extract the final estimate from this observed value (we will discuss those details in the following section). 
 In addition to being easy to state and provided with theoretical guarantees, they are highly practical in nature. 
 We note a following works on algorithmic engineering of practical variants \cite{DBLP:journals/corr/Ertl17,DBLP:conf/edbt/HeuleNH13,xiao2017better}, with actual implementations e.g. in  Algebird \cite{Algebird}, BigQuery \cite{BigQuery}, Apache DataSketch \cite{Apache}, Presto \cite{Presto} and Redis \cite{Redis}.

Despite its simplicity and popularity, \LL and \HLL are exceptionally tough to analyze. 
We note that both papers analyzing \LL and later \HLL use a heavy machinery of tools from analytic combinatorics and complex function analysis to analyze the algorithm guarantees, such as Mellin transform from complex analysis, poissonization for algorithm analysis, and analytical depoissonization (to unpack the main tool used in the paper requires another tens of pages from \cite{szpankowski2011average}). Additionally, all of this is presented in a highly compressed form. Thus the analysis is not easily digestible by a typical computer scientist, and has to be accepted ``as is'' in a black-box manner, without actually unpacking it.

This creates an unsatisfactory situation where one of the most popular and most elegant algorithms for the cardinality estimation problem has to be treated as a black-box from the perspective of its performance guarantees. 
It is an obstacle both in terms of popularization of the \LL and \HLL algorithms, and in terms of scientific progress. Authors note that those algorithms are generally omitted during  majority of theoretical courses on streaming and big data algorithms.

\subsection*{Our contribution.}
Our contribution comes in two factors. First, we observe that a key part of \LL and \HLL algorithms is counting the trailing zeroes in the binary representation of a hash of element.  This random variable is distributed according to geometric distribution. Both \LL and \HLL use the maximal value observed over all elements of the count of trailing zeroes to estimate the cardinality. However, the distribution of many discrete random variables drawn from identical geometric distributions is not distributed according to a geometric distribution. This is unwieldy to handle in the analysis in \cite{flajolet2007hyperloglog}. We propose to replace geometric distribution with Gumbel distribution, which has the following crucial property: 
\\

\emph{If $X_1,\ldots,X_k$ are independent random variables drawn from Gumbel distribution, then $Z = \max(X_1,\ldots,X_k) - \ln(k)$ is also distributed according to the same Gumbel distribution.} 
\\

\noindent This lets us to simplify extraction of value of $k$ from  $\max(X_1,\ldots,X_k)$, since we are always dealing with the same type of error (Gumbel distribution) on top of value of $\ln(k)$.

Our second contribution comes in the form of simple analysis of performance guarantees of the estimation. Instead of analyzing the variance of the estimator itself, we show bounds on intermediate process of maximum of Gumbel random variables. This requires application of some basic probabilistic inequalities and multinomial identities to bound it in the context of stochastic averaging (we discuss this later in the paper).

\section{Related work.}
The key concept used in virtually all cardinality estimation results, can be summarized as follows: given universe $U$ of elements, we start by picking a hash-function. 
Then, given subset $M \subseteq U$ which cardinality we want to estimate, we proceed by applying $h$ to every element of $M$ and operate only on $M' = \{h(x) : x \in M\} \subset [0,1]$. 
The next step is computing an \emph{observable} -- i.e. a quantity that only depends on the underlying set and is independent of replications. Finally step is estimating of the cardinality from the observable.

For example \cite{DBLP:conf/random/Bar-YossefJKST02} uses $h: M \to [0,1]$ and a value $y = \min M' = \min_{x \in M} h(x)$ as an observable. 
We expect $y \sim \frac{1}{n+1}$, thus $\frac{1}{y}-1$ is used as an estimate of cardinality $n$.
However, since we need to overcome the variance, we might need to average over many independent instances of the process, in order to achieve a good estimation.
In this particular example, to get an $(1 + \varepsilon)$ approximation, we need to average over $\bigo(\varepsilon^{-2})$ independent repetitions of the algorithm. 
Therefore, the total memory usage becomes $\bigo(\varepsilon^{-2} \log n)$ bits.

\subsection*{Stochastic averaging.}
\emph{Stochastic averaging} is a technique that in this setting works as follows: instead of processing each of elements in each of $k$ processes independently (which is a bottleneck), we partition our input into $k$ disjoint sub-inputs: $M = M_1 \cup \ldots \cup M_k$, and have each observable follow only processing of a single sub-input. 
This is achieved by picking a second hash function $h': M \to \{1,\ldots,k\}$, and when processing an element $x$, it is assigned to $M_i$ where $i = h'(x)$ is decided solely on hash of $x$. 
Thus we expect each $M_i$ to contain roughly $n/k$ elements. Note that actual number of elements in all $M_i$ follows \emph{multinomial distribution}, and this presents an additional challenge in the analysis.

\subsection*{\LL sketching.}
Consider a following: we hash the elements to bitstrings, that is $h: M \to \{0,1\}^\infty$, and consider the bit-patterns observed. For each element find $\textsf{bit}(x)$ such that $h(x)$ has a prefix $0^{\textsf{bit}(x)}1$.  Value $\textsf{bit}(x)=c$ should be observed once every $\sim 2^{c}$ different hashes, and can be used to estimate the cardinality. The observable used in \LL is the value of $\max_x \textsf{bit}(x)$ among all elements. Since we expect its value to be roughly of order of $\log n$, we maintain the value of $\max \textsf{bit}(x)$ on $\bigo(\log \log n)$ bits. 

A single observable produces a value $t = \max t(x)$. Denote the observables produced over separate sub-streams as $t_1,\ldots,t_k$. We expect the values of $t_i$ to be such that $2^{t_i} \sim n/k$. One can easily show, that for any $t_i$, we have $\E[2^{t_i}] = \infty$, thus arithmetic averaging over $2^{t_i}$ is not a feasible strategy. However, a geometric average works in this setting, and we expect the $k \left( \prod_i 2^{t_i}\right)^{1/k}$ to be an estimate for $n$ (one needs a normalizing constant that depends solely on $k$). The variance analysis shows that the variance of the estimation is roughly $1.3/\sqrt{k}$.

\subsection*{\HLL sketching.}
\HLL (\cite{flajolet2007hyperloglog}) is an improvement over \LL with a following observation, that a \emph{harmonic average} achieves better averaging over \emph{geometric average}. Thus \HLL is constructed by substituting the estimation to be $k^2 \left( \sum_i 2^{-t_i} \right)^{-1}$ with some normalizing constant (depending on $k$). Resulting algorithm has variance which is roughly $1.04/\sqrt{k}$.

In fact it can be shown that the harmonic average is optimal here in this setting: among observables that constitute of taking maximum of a hash function, harmonic average gives is both \emph{maximum likelihood estimator} and \emph{minimum variance estimator} (see e.g. \cite{clifford2012statistical}).  However, those claims are strict only without stochastic averaging.

\section{Preliminaries.}

\paragraph*{Computation model.}
We assume oracle access to a perfect source of randomness, that is a hash function $h: [u] \to \{0,1\}^\infty$. If the sketch demands it, we allow it to access multiple independent such sources, which can be simulated with help of bit or arithmetic operations starting with a single such source a single one. The oracle access is a standard assumption in this line of work (c.f. discussion in \cite{DBLP:journals/corr/abs-2007-08051}) meant to decouple bit-storage of randomness from algorithm analysis. 

Besides that, we assume standard RAM model, with words of size $\log u$ and standard arithmetic operations on those words taking constant time.

\paragraph*{Gumbel distribution.}
We use a following distribution, which originates from \emph{extreme value theory}.

\begin{definition}[Gumbel distribution \cite{gumbel1935valeurs}]
Let $\mathsf{Gumbel}(\mu)$ denote the distribution given by a following CDF:
$$F(x) = e^{-e^{-(x-\mu)}}.$$
Its probability density function is given by
$$f(x) = e^{- e^{-(x-\mu)}} e^{-(x-\mu)}.$$
\end{definition}

\begin{figure}
\begin{tikzpicture}
\begin{axis}[
xmin = -2, xmax = 10,
ymin = 0, ymax = 0.5,
]
\addplot[
domain = 0:10,
samples = 11,
only marks,
black,
] {(1-2^(-x-1))^1 - (1-2^(-x))^1};
\addplot[
domain = 0:10,
samples = 200,
smooth,
black,
] {(1-2^(-x-1))^1 - (1-2^(-x))^1};
\addplot[
domain = 0:10,
samples = 11,
only marks,
black,
] {(1-2^(-x-1))^2 - (1-2^(-x))^2};
\addplot[
domain = 0:10,
samples = 200,
smooth,
black,
] {(1-2^(-x-1))^2 - (1-2^(-x))^2};
\addplot[
domain = 0:10,
samples = 11,
only marks,
black,
] {(1-2^(-x-1))^4 - (1-2^(-x))^4};
\addplot[
domain = 0:10,
samples = 200,
smooth,
black,
] {(1-2^(-x-1))^4 - (1-2^(-x))^4};
\addplot[
domain = 0:10,
samples = 11,
only marks,
black,
] {(1-2^(-x-1))^8 - (1-2^(-x))^8};
\addplot[
domain = 0:10,
samples = 200,
smooth,
black,
] {(1-2^(-x-1))^8 - (1-2^(-x))^8};
\addplot[
domain = 0:10,
samples = 11,
only marks,
black,
] {(1-2^(-x-1))^16 - (1-2^(-x))^16};
\addplot[
domain = 0:10,
samples = 200,
smooth,
black,
] {(1-2^(-x-1))^16 - (1-2^(-x))^16};
\addplot[
domain = 0:10,
samples = 11,
only marks,
black,
] {(1-2^(-x-1))^32 - (1-2^(-x))^32};
\addplot[
domain = 0:10,
samples = 200,
smooth,
black,
] {(1-2^(-x-1))^32 - (1-2^(-x))^32};
\addplot[
domain = 0:10,
samples = 11,
only marks,
black,
] {(1-2^(-x-1))^64 - (1-2^(-x))^64};
\addplot[
domain = 0:10,
samples = 200,
smooth,
black,
] {(1-2^(-x-1))^64 - (1-2^(-x))^64};
\end{axis}
\end{tikzpicture}
\begin{tikzpicture}
\begin{axis}[
xmin = -2, xmax = 10,
ymin = 0, ymax = 0.5,
]
\addplot[
domain = -2:10,
samples = 200,
smooth,
thick,
black,
] {exp(-exp(-x))*exp(-x) )};
\addplot[
domain = -2:10,
samples = 200,
smooth,
thick,
black,
] {exp(-exp(-x + ln(2)))*exp(-x + ln(2)) )};
\addplot[
domain = -2:10,
samples = 200,
smooth,
thick,
black,
] {exp(-exp(-x + ln(4)))*exp(-x + ln(4)) )};
\addplot[
domain = -2:10,
samples = 200,
smooth,
thick,
black,
] {exp(-exp(-x + ln(8)))*exp(-x + ln(8)) )};
\addplot[
domain = -2:10,
samples = 200,
smooth,
thick,
black,
] {exp(-exp(-x + ln(16)))*exp(-x + ln(16)) )};
\addplot[
domain = -2:10,
samples = 200,
smooth,
thick,
black,
] {exp(-exp(-x + ln(32)))*exp(-x + ln(32)) )};
\addplot[
domain = -2:10,
samples = 200,
smooth,
thick,
black,
] {exp(-exp(-x + ln(64)))*exp(-x + ln(64)) )};
\end{axis}
\end{tikzpicture}
\caption{Distribution of $\max\{X_1,\ldots,X_k\}$ for $k\in \{1,2,4,8,16,32,64\}$ where $X_i$ iid random variables distributed according to discrete Geometric distribution (on the left) and Gumbel distribution (on the right). Discrete distribution given by $f_k(x) = (1-2^{-x-1})^k - (1-2^{-x})^k$ is drawn with continuous intermediate values for smooth drawing.}
\end{figure}
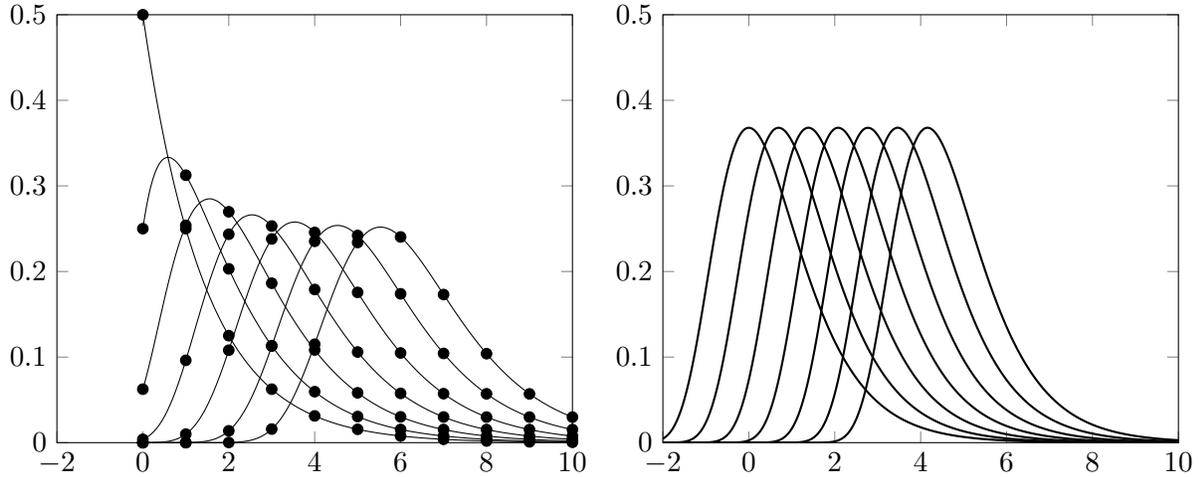

We note that when $x \to \infty$, then $f(x) \approx e^{-(x - \mu)}$, thus the Gumbel distribution has the exponential tail on the positive side. The distribution has a doubly-exponential tail when $x \to -\infty$.

We also have the following basic properties when $X \sim \mathsf{Gumbel}(\mu)$ (c.f. \cite{gumbel1935valeurs}):
\begin{equation}
\label{eq:expandvar}
\E[X-\mu] =  \gamma \approx 0.5772,\quad\quad\Var[X] = \frac{\pi^2}{6} \approx 1.6449.
\end{equation}
and
\begin{equation}
\label{eq:expexp}
\E[e^{-X}] = e^{-\mu} \int_{-\infty}^{\infty} e^{- e^{-x}} e^{-2x}  dx =  e^{-\mu},
\end{equation}
\begin{equation}
\label{eq:expvar}
\Var[e^{-X}] = \E[(e^{-X})^2] - e^{-2\mu} = e^{-2\mu} \int_{-\infty}^{\infty} e^{- e^{-x}} e^{-3x}  dx - e^{-2\mu}= e^{-2\mu}.
\end{equation}

\begin{property}[Sampling from Gumbel distribution.]
If $t \in [0,1]$ is drawn uniformly at random, then $X = -\ln(-\ln t) + \mu$ has the distribution $\mathsf{Gumbel}(\mu)$.
\end{property}

The following property is a key property used in our algorithm analysis. It essentially states that Gumbel distribution is invariant under taking the maximum of independent samples (up to normalization).\footnote{In fact, the Fisher–Tippett–Gnedenko theorem (c.f. \cite{de2007extreme}) states, that for any distribution $\mathcal{D}$, if for some $a_n, b_n$ the limit $\lim_{n \to \infty} ( \frac{\max(X_1,\ldots,X_n)-b_n}{a_n})$ converges to some non-degenerate distribution, where $X_1,\ldots X_n \sim \mathcal{D}$ (and are independent), then it converges to one of three possible distribution families: a Fréchet distribution, a Weibull distribution or a Gumbel distribution. Thus, those three distributions can be viewed as a counterpart to normal distribution, wrt to taking maximum (instead of repeated additions).}
\begin{property}
\label{maxproperty}
If $x_1,x_2,\ldots,x_n \sim \mathsf{Gumbel}(0)$ are independent random variables, then for $Z = \max(x_1,\ldots,x_n)$ we have $Z \sim \mathsf{Gumbel}(\ln n)$.
\end{property}
\begin{proof}
$$\Pr(Z < x) = \prod_{i} \Pr(x_i < x) = (e^{e^{-x}})^n = e^{e^{-x+\ln n}}.\qedhere$$
\end{proof}

\paragraph*{Multinomial distribution.}
We now discuss the multinomial distribution and its role in analyzing stochastic averaging.
\begin{definition}
We say that $X_1,\ldots,X_k$ are distributed according to $\textsf{Multinomial}(n;p_1,\ldots,p_k)$ distribution for some $\sum_i p_i = 1$, if, for any $n_1+\ldots+n_k = n$ there is
$$\Pr[X_1 = n_1 \wedge \ldots \wedge X_k = n_k] =  {n \choose n_1,\ldots,n_k} p_1^{n_1} \ldots p_k^{n_k}.$$
\end{definition}

Consider a process of distributing $n$ identical balls to $k$ urns, where each the probability for any ball to land in urn $i$ is $p_i$, fully independently between balls. Then the numbers of total balls in each urn $X_1,\ldots,X_k$ follows $\textsf{Multinomial}(n;p_1,\ldots,p_k)$ distribution.

For our purposes we are interested in the following: let $f$ be some real-value function. Lets say that we have a stochastic process of estimating cardinality in a stream, that is if $n$ distinct elements appear, the process outputs a value that is concentrated around its expected value $f(n)$. Now, we apply stochastic averaging, by splitting the stream into sub-streams, and feed each sub-stream to estimation process separately, say $n_i$ going into sub-stream $i$. We can look at the following random variables:
$$S_n = \E[ \sum_i f(n_i)] \qquad\qquad \textrm{and} \qquad\qquad P_n = \E[ \prod_i f(n_i)].$$
We expect $S_n \approx k f(n/k)$ and $P_n \approx f(n/k)^k$. Deriving actual concentration bounds for specifically chosen functions $f$ gives us insight on how well harmonic average or geometric average performs when concentrating cardinality estimation processes under stochastic averaging.

The analysis of stochastic averaging for a \emph{generic} function $f$ (under some sanity constraints) has been done in \cite{clifford2012statistical}. We actually derive a stronger set of bounds for very specific functions: $f(x) = \frac{1}{x+1}$ and $f(x) = \ln(x+1)$.

\section{Geometric average estimation.}

Following algorithm shows that if we are fine with slower updates, then Gumbel distribution plays nicely into estimating cardinality. The main idea is just to hash each element into a real-value distributed according to Gumbel distribution, and take maximum across all values. 

\SetKwProg{procedure}{Procedure}{}{}
\begin{algorithm}[H]
\label{alg1}
\caption{Cardinality estimation using Gumbel distribution.}
	\DontPrintSemicolon
%	\KwData{}
%    \KwResult{}
\procedure{\textsc{Init()}}{
	pick $h_1,\ldots,h_k : U \to [0,1]$ as independent hash functions\;
	$X_1 \gets -\infty,\ldots,X_k \gets -\infty$\;
}
\procedure{\textsc{Update($x$)}}{
	\For{ $1 \le i \le k$}{
		$v \gets -\ln(-\ln h_i(x))$ \tcp{Gumbel(0)  RV}
		$X_i \gets \max(v, X_i)$\;
	}
}
\procedure{\textsc{GeometricEstimate()}}{
	\Return $Z = \exp(-\gamma + \frac{1}{k} \sum_i X_i)\;$
	%\Return $Z = k \cdot \left(\sum_i \exp(-X_i)\right)^{-1}$\;
}

\end{algorithm}

\begin{theorem}
\label{firstTheorem}
Applied to a stream of $n$ distinct elements, Algorithm~\ref{alg1} outputs $Z$ such that $|Z - n| \le n \cdot (\pi k^{-1/2} + \bigo(k^{-1}))$ holds with constant probability $5/6$. It uses $k$ real-value registers and spends $\bigo(k)$ operations per single processed element of the input.
\end{theorem}
Thus, setting $k = \varepsilon^{-2}$ gives a constant probability for Algorithm~\ref{alg1} outputting a $(1 + \varepsilon)$-multiplicative estimation of cardinality.

\begin{proof}
We analyze Algorithm~\ref{alg1} after processing stream of $n$ distinct elements. For each $X_i$, its value is a maximum of $n$ random variables drawn from $ \textsf{Gumbel}(0)$ distribution, so by Property~\ref{maxproperty} we have that $X_i \sim \textsf{Gumbel}(\ln n)$. Moreover, repeated occurrences of elements in the stream do not change the state of the algorithm.

By Equation~\eqref{eq:expandvar}
$$\E[X_i] = \gamma + \ln n\qquad \textrm{and} \qquad \Var[X_i] = \frac{\pi^2}{6}.$$
Thus for $X =  \sum_i X_i$ there is $\E[X] = k \gamma + k \ln n$ and $\Var[X] =  k \frac{\pi^2}{6}$. By Chebyshev's inequality: 
$$\Pr( |X- \E[X]| \ge \pi  \sqrt{k}) \le 1/6.$$ 
Since $Z = \exp(-\gamma + X/k)$, we have that (with probability at least $5/6$)
$$n \cdot\exp\left(1-\pi k^{-1/2}\right) \le Z \le n \cdot\exp\left(1 + \pi k^{-1/2}\right).\qedhere$$.
\end{proof}

\subsection{Stochastic averaging.}
We refine Algorithm~\ref{alg1} with stochastic averaging. Application of the technique is straightforward, but we need to take care of initialization of $X_i$ registers.

\begin{algorithm}[H]
\label{alg2}
\caption{Cardinality estimation using Gumbel distribution and stochastic averaging.}
	\DontPrintSemicolon
%	\KwData{}
%    \KwResult{}
\procedure{\textsc{Init()}}{
	pick $h : U \to \{1,\ldots,k\}$ and $r : U \to [0,1]$ as independent hash functions\;
	\For{ $1 \le i \le m$}{
		$X_i \gets -\ln(-\ln u_i)$ where $u_i$ is picked uniformly from $[0,1]$. \tcp{Gumbel(0)  RV}
	}
}
\procedure{\textsc{Update($x$)}}{
	$c \gets h(x)$\;
	$v \gets -\ln(-\ln r(x))$ \tcp{Gumbel(0)  RV}
	$X_c \gets \max(v, X_c)$\;
	
}

\procedure{\textsc{GeometricEstimate()}}{
	\Return $Z = k \cdot \exp(-\gamma + \frac{1}{k} \sum_i X_i)$
	%\Return $Z = k^2 \cdot (\sum_i \exp(-X_i))^{-1} - 1$\;
}
\end{algorithm}
\begin{theorem}
\label{th:42}
Applied to a stream of $n$ distinct elements, Algorithm~\ref{alg2} outputs $Z$ such that $|Z - n|  = \pi n k^{-1/2}+\bigo(k)$ holds with probability $2/3$. It uses $k$ real-value registers and spends constant number of operations per single processed element of the input.
\end{theorem}
Thus, setting $k = \varepsilon^{-2}$ gives a constant probability for Algorithm~\ref{alg2} outputting a $(1 + \varepsilon)$-multiplicative estimation of cardinality, assuming $n \ge k^{3/2} = \varepsilon^{-3}$.

\begin{proof}
We analyze Algorithm~\ref{alg2} after processing stream $S$ of $n$ distinct elements.
Let $n_1, \ldots, n_k$ be the respective numbers of unique items hashed by $h$ into buckets $\{1,\ldots,k\}$ respectively. It follows that $n_1,\ldots,n_k \sim \mathsf{Multinomial}(n;\frac1k,\ldots,\frac1k)$. For each $X_i$, its value is a maximum of $n_i+1$ random variables drawn from $ \textsf{Gumbel}(0)$ distribution (taking into account $n_i$ updates to its value and initial value). Thus conditioned on specific values of $n_1,\ldots,n_k$, we have that $X_i$ follows the Gumbel distribution. More specifically $X_i | n_1,\ldots,n_k \sim \mathsf{Gumbel}(\ln (n_i+1)).$ We also observe, that for $i\not=j$, $X_i | n_1,\ldots,n_k$ and $X_j | n_1,\ldots,n_k$ are independent random variables.

Denote $X = \sum_i X_i$ and $Y = \sum_i \ln(n_i+1)$. We split our analysis of $X$ into two parts. First, almost identical analysis to one from Theorem~\ref{firstTheorem} follows:
$$\E[X_i\ |\ n_1,\ldots,n_k] = \gamma + \ln(n_i+1)\qquad\textrm{and}\qquad\Var[X_i\ |\ n_1,\ldots,n_k]=\frac{\pi^2}{6}$$
thus
$$\Pr( |X - (k \gamma + Y) | \ge \pi \sqrt{k}\ |\ n_1,\ldots,n_k ) \le 1/6.$$
We can drop the conditional part and write
\begin{equation}
\label{eqfirstbound}
\Pr( |X - (k \gamma + Y) | \ge \pi \sqrt{k}) \le 1/6.
\end{equation}

We now show concentration of the second part of sum. First, by convexity we get.

\begin{equation}
\label{eqsecondbound}
Y = \sum_i \ln(n_i+1) \le k \ln(n/k+1).
\end{equation}

By Lemma~\ref{lemmultinomial} we get that
\begin{equation}
\label{eqthirdbound}
\Pr[Y \ge k \ln(n/k) - \ln 6] \ge 5/6.
\end{equation}

Combining Equations~\eqref{eqfirstbound}, \eqref{eqsecondbound} and \eqref{eqthirdbound} we reach that the following bound holds with probability at least $2/3$:
$$  k\gamma + (k \ln(n/k) - \ln 6) - \pi \sqrt{k} \le X \le k\gamma + k \ln((n+k)/k) + \pi \sqrt{k} $$
or equivalently, since $Z = k \exp(-\gamma + X/k)$
$$ n \cdot ( 1 -\pi k^{-1/2}  - \bigo(k^{-1})) \le Z \le (n+k) \cdot (1 + \pi  k^{-1/2} + \bigo(k^{-1})).\qedhere$$
\end{proof}

\begin{lemma}
\label{lemmultinomial}
Let $n_1,\ldots,n_k \sim \mathsf{Multinomial}(n;1/k,\ldots,1/k)$ and let $Y = \sum_i \ln(n_i+1)$. Then $Y \ge k \ln(n/k) - t$ with probability at least  $1-e^{-t}$.
\end{lemma}
\begin{proof}
Consider $\E[e^{-Y}]$. We have
\begin{align*}
\E\limits_{\substack{n_1,..,n_k\sim\\  \mathsf{Multinomial}}}[e^{-Y}] &= \E\limits_{\substack{n_1,..,n_k\sim\\  \mathsf{Multinomial}}} \left[ \prod_i \frac{1}{n_i+1} \right]\\
&= \sum_{i_1+\ldots+i_k = n} \Pr[n_1 = i_1 \wedge \ldots \wedge n_k = i_k] \prod_i \frac{1}{i_i+1}\\
&= \sum_{i_1+\ldots+i_k = n} k^{-n} {n \choose i_1,\ldots,i_k} \prod_i \frac{1}{i_i+1}\\
&= k^{-n}  \sum_{i_1+\ldots+i_k = n} \frac{n!}{(i_1+1)! \cdot \ldots \cdot (i_k+1)!}\\
&= k^{-n}  \sum_{i_1+\ldots+i_k = n}  {n+k \choose i_1+1,\ldots,i_k+1} \frac{n!}{(n+k)!} \\
&\le k^{-n}  k^{n+k} \frac{n!}{(n+k)!} \\
&\le \left(\frac{k}{n}\right)^k
\end{align*}

Thus, for any $t>0$, by Markov's inequality
\begin{align*}
\Pr[Y \le k \ln(n/k) - t] &= \Pr[e^{-Y} \ge e^{t - k \ln(n/k)}]\\ 
&= \Pr[e^{-Y} \ge e^t \cdot \E[e^{-Y}]]\\ 
&\le e^{-t}.\tag*{\qedhere}
\end{align*}
\end{proof}

\subsection{Discretization.}

Presented sketches use $k$ real-value registers, which is in disadvantage when compared with \LL and \HLL, where only $k$ integers are used, each taking $\bigo(\log \log n)$ bits. We now discuss how to reduce the memory footprint of the algorithms. 

\paragraph*{Simple rounding.} First we note that rounding the registers to nearest multiplicity of $\varepsilon$ for some $\varepsilon>0$ introduces at most $\exp(1+\varepsilon) = 1 + \varepsilon + \bigo(\varepsilon^2)$ multiplicative distortion, both with the estimation procedure $\textsf{GeometricEstimate()}$ from Algorithm~\ref{alg1} and \ref{alg2} and with the estimation procedure $\textsf{HarmonicEstimate()}$ from Algorithm~\ref{alg3} and \ref{alg4} (see Appendix). For example, for \ref{alg1}, we have, assuming $X'_i$ are rounded registers: $|X'_i - X_i| \le \varepsilon$, and so for $Z' = \exp(-\gamma + \frac{1}{k} \sum_i X'_i)$ there is $\frac{Z'}{Z} = \exp(\frac{1}{k}\sum_i(X'_i - X_i))$, so $\exp(-\varepsilon) \le \frac{Z'}{Z} \le \exp(\varepsilon)$. Since each register stores w.h.p. values of magnitude $2 \log n$, it can be implemented on integer registers using $\bigo(\log \frac{\log n}{\varepsilon}) = \bigo(\log \log n + \log \varepsilon^{-1})$ bits.

\paragraph*{Randomized rounding.} We now show how to eliminate the $\log \varepsilon^{-1}$ term. We define the following \emph{shift-rounding}, for shift value $c \in [0,1)$:
$$f_c(x) \defeq \lfloor x+c \rfloor-c.$$
We note two key properties: 
\begin{enumerate} 
\item shift-rounding commutes with maximum, that is, for any $x_1,\ldots,x_k$, we have $\max(f_c(x_1), \ldots, f_c(x_k)) = f_c( \max(x_1,\ldots,x_k))$,
\item If $c \sim U[0,1]$, then $f_c(x) \sim U[x-1,x]$, where $U[a,b]$ denotes uniform distribution on range $[a,b]$.
\end{enumerate}

We thus show how to adapt the Algorithm \ref{alg2} using shift-rounding.

\begin{algorithm}[H]
\label{alg_rounding}
\caption{Algorithm \ref{alg2} with shift-rounding.}
	\DontPrintSemicolon
%	\KwData{}
%    \KwResult{}
\procedure{\textsc{Init()}}{
	pick $h : U \to \{1,\ldots,k\}$ and $r : U \to [0,1]$ as independent hash functions\;
	\For{ $1 \le i \le m$}{
		$c_i$ is picked uniformly from $[0,1]$\;
		$X_i \gets \lfloor-\ln(-\ln u_i) + c_i\rfloor - c_i$\;
		 where $u_i$ is picked uniformly from $[0,1]$. \tcp{Gumbel(0)  RV}
	}
}
\procedure{\textsc{Update($x$)}}{
	\For{ $1 \le i \le k$}{
		$v \gets \lfloor -\ln(-\ln h_i(x)) + c_i \rfloor-c_i$\;
		$X'_i \gets \max(v, X'_i)$\;
	}
}
\procedure{\textsc{GeometricEstimate()}}{
	\Return $Z = k \exp(-\gamma + \frac{1}{2} + \frac{1}{k} \sum_i X'_i)\;$
}

\end{algorithm}

The analysis of Algorithm~\ref{alg_rounding} comes from following invariant: if Algorithms~\ref{alg_rounding} and \ref{alg2} are run side-by-side on the same input stream, at any given moment there is $X'_i = f_{c_i}(X_i)$. Thus, we have the following $X'_i \sim \textsf{Gumbel}(\ln n_i) - U[0,1]$. So $\E[X'_i] = \gamma - \frac12 + \ln n_i$, and $\Var[X'_i] = \frac{\pi^2}{6} + \frac14$. Additionally, $X'_i$ are independent as $X_i$ were independent. Thus an equivalent of Theorem \ref{firstTheorem} applies to Algorithm~\ref{th:42} with slightly worse constants.

\begin{theorem}
Applied to a stream of $n$ distinct elements, Algorithm~\ref{alg_rounding} outputs $Z$ such that $|Z - n|  = \bigo(n k^{-1/2}+k)$ holds with probability $2/3$. It uses $k$ integer registers of size $\bigo(\log \log n)$ bits each and spends constant number of operations per single processed element of the input.
\end{theorem}

We note that each $X'_i$ takes values only from set $\mathbb{Z}-c_i$ of magnitude at most $2 \log n$, it can be stored using $\bigo(\log \log n)$ bits. Values of $c_i$ do not need to be stored explicitly, as those can be extracted by picking a hash function $c: \{1,\ldots,k\} \to [0,1]$ and setting $c_i = c(i)$.

We note that analogous adaptation is straightforward to other algorithms presented in this paper.

\bibliography{bib}

%\pagenumbering{roman}
\appendix

\newpage

\section{Harmonic average estimation.}

\begin{algorithm}[H]
\label{alg3}
\caption{Improved estimation for Algorithm~\ref{alg1}.}
	\DontPrintSemicolon
\textbf{\textrm{Procedure}} \textsc{Init()} \tcp{identical as in Algorithm~\ref{alg1}}
\textbf{\textrm{Update}} \textsc{Update($x$)} \tcp{identical as in Algorithm~\ref{alg1}}
\procedure{\textsc{HarmonicEstimate()}}{
	%\Return $Z = \exp(-\gamma + \frac{1}{k} \sum_i X_i)\;$
	\Return $Z = k \cdot \left(\sum_i \exp(-X_i)\right)^{-1}$\;
}
\end{algorithm}
\begin{theorem}
\label{thirdTheorem}
Applied to a stream of $n$ distinct elements,  Algorithm~\ref{alg3} outputs $Z$ such that $|Z - n| \le n \cdot (2 k^{-1/2} + \bigo(k^{-1}))$ holds with constant probability $3/4$. It uses $k$ real-value registers and spends $\bigo(m)$ operations per single processed element of the input.
\end{theorem}
Thus, setting $k = \varepsilon^{-2}$ gives a constant probability for Algorithm~\ref{alg3} outputting a $(1 + \varepsilon)$-multiplicative estimation of cardinality.

\begin{proof}
We analyze Algorithm~\ref{alg3} after processing stream of $n$ distinct elements. For each $X_i$, its value is a maximum of $n$ random variables drawn from $ \textsf{Gumbel}(0)$ distribution, so by Property~\ref{maxproperty} we have that $X_i \sim \textsf{Gumbel}(\ln n)$. Moreover, repeated occurrences of elements in the stream do not change the state of the algorithm.

Denote $U_i = e^{-X_i}$. By Equations \eqref{eq:expexp} and \eqref{eq:expvar}
we have $\E[U_i] = \frac{1}{n}$ and $\Var[U_i] = \frac{1}{n^2}$. Denoting $U = \sum_i U_i$, we have $\E[U] = \frac{k}{n}$ and $\Var[U] = \frac{k}{n^2}$. Thus by standard application of Chebyshev's inequality
$$\Pr\Big[ |U - \frac{k}{n}| \le 2 \frac{\sqrt{k}}{n} \Big] \le \frac{1}{4}.$$
Taking into account that $Z = \frac{k}{U}$ we reach the claim.
\end{proof}

\subsection{Stochastic averaging.}
\begin{algorithm}[H]
\label{alg4}
\caption{Improved estimation for Algorithm~\ref{alg2}.}
	\DontPrintSemicolon
\textbf{\textrm{Procedure}} \textsc{Init()} \tcp{identical as in Algorithm~\ref{alg2}}
\textbf{\textrm{Update}} \textsc{Update($x$)} \tcp{identical as in Algorithm~\ref{alg2}}
\procedure{\textsc{HarmonicEstimate()}}{
	%\Return $Z = k \cdot \exp(-\gamma + \frac{1}{k} \sum_i X_i)$
	\Return $Z = k^2 \cdot (\sum_i \exp(-X_i))^{-1} - 1$\;
}
\end{algorithm}

\begin{theorem}
\label{mainresult}
Applied to a stream of $n$ distinct elements, Algorithm~\ref{alg4} outputs $Z$ such that $|Z - n|  = \bigo(n k^{-1/2} + n\exp(-n/k))$ holds with constant probability $3/4$. It uses $k$ real-value registers and spends constant number of operations per single processed element of the input.
\end{theorem}
Thus, setting $k = \varepsilon^{-2}$ gives a constant probability for Algorithm~\ref{alg4} outputting a $(1 + \varepsilon)$-multiplicative estimation of cardinality, assuming $n \ge k \log k = \varepsilon^{-2} \log \varepsilon^{-1}$.

\begin{proof}
We analyze Algorithm~\ref{alg2} after processing stream $S$ of $n$ distinct elements.
Let $n_1, \ldots, n_k$ be the respective numbers of unique items hashed by $h$ into buckets $\{1,\ldots,k\}$ respectively. It follows that $n_1,\ldots,n_k \sim \mathsf{Multinomial}(n;\frac1k,\ldots,\frac1k)$. For each $X_i$, its value is a maximum of $n_i+1$ random variables drawn from $ \textsf{Gumbel}(0)$ distribution (taking into account $n_i$ updates to its value and initial value). Thus conditioned on specific values of $n_1,\ldots,n_k$, we have that $X_i$ follows the Gumbel distribution. More specifically $X_i | n_1,\ldots,n_k \sim \mathsf{Gumbel}(\ln (n_i+1)).$ We also observe, that for $i\not=j$, $X_i | n_1,\ldots,n_k$ and $X_j | n_1,\ldots,n_k$ are independent random variables.

Denote $U_i = e^{-X_i}$ and $U = \sum_i U_i$.  We derive following bound on conditional expected value

\begin{align*}
\E[U\ |\ n_1,\ldots,n_k] &= \sum_i \E[U_i\ |\ n_1,\ldots,n_k ]\\
&= \sum_i \exp(-\ln(n_i+1)) \tag*{\text{(by Equation \eqref{eq:expexp})}}\\
&= \sum_i \frac{1}{n_i+1},
\end{align*}

and bound on conditional variance

\begin{align*}
\Var[U\ |\ n_1,\ldots,n_k] &= \sum_i \Var[U_i\ |\ n_1,\ldots,n_k ] \tag*{(independence)}\\
&=\sum_i \exp(-2\ln(n_i+1)) \tag*{(by Equation \eqref{eq:expvar})}\\
%&= \sum_i \frac{1}{(n_i+1)^2}\\
&\le \sum_i \frac{2}{(n_i+1)(n_i+2)}.
\end{align*}

Denoting $V = \sum_i \frac{1}{n_i+1}$  and $W = \sum_i \frac{2}{(n_i+1)(n_i+2)}$. Also, let $\beta_k = (1-1/k)^k \le 1/e$ be a constant dependent only on $k$.

We have
\begin{align*}
\E[U] &= \E\limits_{\substack{n_1,..,n_k\sim\\  \mathsf{Multinomial}}}[ \E[U\ |\ n_1,\ldots,n_k] ]\\
&= \E\limits_{\substack{n_1,..,n_k\sim\\  \mathsf{Multinomial}}}[ V ] \tag*{(definition of $V$)}\\
&= \frac{k^2}{n+1}(1-\beta_k^{\frac{n+1}{k}}), \tag*{(by Lemma~\ref{lemmultinomial2})}
\end{align*}
and
\begin{align*}
\Var[U] &= \E\limits_{\substack{n_1,..,n_k\sim\\  \mathsf{Multinomial}}}[ \Var[U | n_1,\ldots,n_k]] + \Var\limits_{\substack{n_1,..,n_k\sim\\  \mathsf{Multinomial}}}[ \E[ U | n_1,\ldots,n_k ] ] \tag*{(Law of Total Variance)}\\
&\le \E\limits_{\substack{n_1,..,n_k\sim\\  \mathsf{Multinomial}}}[ W ] + \Var\limits_{\substack{n_1,..,n_k\sim\\  \mathsf{Multinomial}}}[ V ] \tag*{(definition of $W$ and $V$)}\\
&\le 3 \frac{k^3}{(n+1)^2} + 2 \beta_k^{\frac{n+1}{k}}\frac{k^4}{(n+1)^2}. \tag*{(Lemmas~\ref{lemmultinomial3} and \ref{lemmultinomial4})}
\end{align*}
Applying Chebyshev's inequality, we have with constant probability at least $3/4$:
$$\Big|U - \frac{k^2}{n+1}\Big| \le \frac{k^2}{n+1}\left( \beta_k^{\frac{n+1}{k}} + \bigo(k^{-1/2})+ \bigo(\beta_k^{\frac{n+1}{2k}})\right).$$

The claim follows from the fact that $Z = k^2 U^{-1} - 1$.
\end{proof}

\begin{lemma}
\label{lemmultinomial2}
Let $n_1,\ldots,n_k \sim \mathsf{Multinomial}(n;\frac1k,\ldots,\frac1k)$ and let $V = \sum_i \frac{1}{n_i+1}$. Then $\E[V] =\frac{k^2}{n+1}(1-\beta_k^{\frac{n+1}{k}})$.
\end{lemma}
\begin{proof}
Consider the following
\begin{align*}
\E[ V ] &= k \E\left[ \frac{1}{n_1+1} \right]\\
&= k \sum_{i=0}^n \Pr[n_1 = i] \frac{1}{i+1}\\
&= k \sum_{i=0}^n {n \choose i} \frac{(k-1)^{n-i}}{k^n} \cdot \frac{1}{i+1}\\
&= \frac{k}{n+1} \sum_{i=0}^n {n+1 \choose i+1} \frac{(k-1)^{n-i}}{k^n}\\
&= \frac{k}{n+1} \cdot \frac{(k-1+1)^n-(k-1)^n}{k^n}\\
&= \frac{k^2}{n+1} \cdot \left(1 - \left(1-\frac1k\right)^{n+1}\right). \tag*{\qedhere}
\end{align*}
\end{proof}

\begin{lemma}
\label{lemmultinomial3}
Let $n_1,\ldots,n_k \sim \mathsf{Multinomial}(n;\frac1k,\ldots,\frac1k)$ and let $W = \sum_i \frac{2}{(n_i+1)(n_i+2)}$. Then $\E[W] \le \frac{2k^3}{(n+1)^2}$.
\end{lemma}
\begin{proof}
Consider the following
\begin{align*}
\E[W] &= k \E\left[\frac{2}{(n_1+1)(n_1+2)}\right]\\
&= k \sum_{i=0}^n \Pr[n_1 = i] \frac{2}{(i+1)(i+2)}\\
&= k \sum_{i=0}^n {n \choose i} \frac{(k-1)^{n-i}}{k^n} \cdot \frac{2}{(i+1)(i+2)}\\
&=\frac{2k}{(n+1)(n+2)} \sum_{i=0}^n {n+2 \choose i+2} \frac{(k-1)^{n-i}}{k^n}\\
&\le \frac{2k}{(n+1)(n+2)} \cdot \frac{(k-1+1)^{n+2}}{k^n}\\
&=\frac{2k^3}{(n+1)(n+2)}\\
&\le \frac{2k^3}{(n+1)^2}.\tag*{\qedhere}
\end{align*}
\end{proof}

\begin{lemma}
\label{lemmultinomial4}
Let $n_1,\ldots,n_k \sim \mathsf{Multinomial}(n;1/k,\ldots,1/k)$ and let $V = \sum_i \frac{1}{n_i+1}$. Then $\Var[V] \le  \frac{k^3}{(n+1)^2} +  \frac{k^4}{(n+1)^2} \cdot 2 \beta_k^{\frac{n+1}{k}}$, where $\beta_k \approx 1/e$.
\end{lemma}
\begin{proof}
Consider the following
$\E[V^2] = \E[\sum_i  \frac{1}{(n_i+1)^2} ] +  \E[\sum_{i \not= j} \frac{1}{(n_i+1)(n_j+1)}]$.
Since $\frac{1}{(n_i+1)^2} \le \frac{2}{(n_i+1)(n_i+2)}$, by Lemma~\ref{lemmultinomial3} first term satisfies 
$$\E[\sum_i  \frac{1}{(n_i+1)^2} ] \le \E[W] \le \frac{2k^3}{(n+1)^2}.$$ For the second term, consider
\begin{align*}
\E[\sum_{i \not= j} \frac{1}{(n_i+1)(n_j+1)}] &\le k(k-1) \E\left[\frac{1}{(n_1+1)(n_2+1)}\right]\\
&= k(k-1) \sum_{i,j \ge 0} \Pr[n_1 = i \wedge n_2 = j] \frac{1}{(i+1)(j+1)}\\
&= k(k-1) \sum_{i,j \ge 0} {n \choose i,j,n-i-j} \frac{(k-2)^{n-i-j}}{k^n} \frac{1}{(i+1)(j+1)}\\
&= \frac{k(k-1)}{(n+1)(n+2)} \sum_{i,j \ge 0} {n+2 \choose i+1,j+1,n-i-j} \frac{(k-2)^{n-i-j}}{k^n}\\
&\le \frac{k(k-1)}{(n+1)(n+2)} \sum_{i,j \ge 0} {n+2 \choose i,j,n+2-i-j} \frac{(k-2)^{n+2-i-j}}{k^n}\\
&= \frac{k^3(k-1)}{(n+1)(n+2)} \le \frac{k^3(k-1)}{(n+1)^2}.
\end{align*}
Additionally, following bound holds
\begin{align*}
\Var[V] &= \E[V^2] - (\E[V])^2 \\
&\le \frac{2k^3}{(n+1)^2} + \frac{k^3(k-1)}{(n+1)^2} - \frac{k^4}{(n+1)^2}\left(1 - 2 \beta_k^{\frac{n+1}{k}}\right)\\
&\le \frac{k^3}{(n+1)^2} +  \frac{k^4}{(n+1)^2} \cdot 2 \beta_k^{\frac{n+1}{k}}. \tag*{\qedhere}
\end{align*}
\end{proof}

\subsection{Discretization.}
We note that  techniques used in Algorithm \ref{alg_rounding} can be applied with harmonic estimation, leading to a following algorithm.

\begin{algorithm}[H]
\label{alg_rounding_harmonic}
\caption{Improved estimation for Algorithm~\ref{alg_rounding}.}
	\DontPrintSemicolon
\textbf{\textrm{Procedure}} \textsc{Init()} \tcp{identical as in Algorithm~\ref{alg_rounding}}
\textbf{\textrm{Update}} \textsc{Update($x$)} \tcp{identical as in Algorithm~\ref{alg_rounding}}
\procedure{\textsc{HarmonicEstimate()}}{
	%\Return $Z = k \cdot \exp(-\gamma + \frac{1}{k} \sum_i X_i)$
	\Return $Z = k\cdot (\frac{1}{2} + \frac{1}{k}\sum_i \exp(-X_i))^{-1} - 1$\;
}
\end{algorithm}

\begin{theorem}
Applied to a stream of $n$ distinct elements, Algorithm~\ref{alg_rounding_harmonic} outputs $Z$ such that $|Z - n|  = \bigo(n k^{-1/2} + n\exp(-n/k))$ holds with probability $2/3$. It uses $k$ integer registers of size $\bigo(\log \log n)$ bits each and spends constant number of operations per single processed element of the input.
\end{theorem}

\end{document}